\documentclass[a4paper,UKenglish,cleveref, autoref, thm-restate]{lipics-v2021}
\usepackage{amsmath,amsthm,amsfonts}
\usepackage{xcolor}
\usepackage{xspace}
\usepackage{hyperref}
\usepackage{cleveref}
\usepackage{tcolorbox}
\usepackage{todonotes}

\usepackage{tikz,tkz-graph,tkz-berge}
\usetikzlibrary{decorations.pathreplacing}
\usetikzlibrary{calc} 
\usetikzlibrary{graphs}
\usetikzlibrary{graphs.standard}
\usetikzlibrary{shapes,snakes}
\usetikzlibrary{fit,positioning}

\newcommand\fauxsc[1]{\fauxschelper#1 \relax\relax}
\def\fauxschelper#1 #2\relax{%
  \fauxschelphelp#1\relax\relax%
  \if\relax#2\relax\else\ \fauxschelper#2\relax\fi%
}
\def\Hscale{.85}\def\Vscale{.72}\def\Cscale{1.10}
\def\fauxschelphelp#1#2\relax{%
  \ifnum`#1>``\ifnum`#1<`\{\scalebox{\Hscale}[\Vscale]{\uppercase{#1}}\else%
    \scalebox{\Cscale}[1]{#1}\fi\else\scalebox{\Cscale}[1]{#1}\fi%
  \ifx\relax#2\relax\else\fauxschelphelp#2\relax\fi}

\newcommand{\set}[1]{\{#1\}}
\newcommand{\modulus}[1]{\vert #1 \vert}

\newcommand{\calG}{\mathcal{G}}

\newcommand{\cost}{\mathit{cost}}

\newcommand{\Mod}[1]{\ \mathrm{mod}\ #1}

\newcommand{\wonderful}{{\fauxsc{Wonderful-Partition}}}
\newcommand{\wonderfulIntervals}{{\fauxsc{Wonderful-Partition-Intervals}}}

\newcommand{\minEg}{{\fauxsc{Min-Eg}}}

\title{Solving Woeginger's Hiking Problem: Wonderful Partitions in Anonymous Hedonic Games}

\titlerunning{Solving Woeginger's Hiking Problem} %TODO optional, please use if title is longer than one line

\author{Andrei Constantinescu\footnote{Corresponding author.}}{ETH Zürich, Zürich, Switzerland \and \url{https://andrei1998.github.io/}}{aconstantine@ethz.ch}{https://orcid.org/0009-0005-1708-9376}{}

\author{Pascal Lenzner}{Hasso Plattner Institute, University of Potsdam, Potsdam, Germany}{pascal.lenzner@hpi.de}{https://orcid.org/0000-0002-3010-1019}{}

\author{Rebecca Reiffenhäuser}{University of Amsterdam, Amsterdam, The Netherlands}{r.e.m.reiffenhauser@uva.nl}{https://orcid.org/0000-0002-0959-2589}{}

\author{Daniel Schmand}{University of Bremen, Bremen, Germany}{schmand@uni-bremen.de}{https://orcid.org/0000-0001-7776-3426}{}

\author{Giovanna Varricchio}{University of Calabria, Rende, Italy}{giovanna.varricchio@unical.it}{https://orcid.org/0000-0001-6839-8551}{The author is grateful for the support by the PNRR MUR project FAIR - Future AI Research (PE00000013), and the DFG, German Research Foundation, grant (Ho 3831/5-1).}

\authorrunning{A. Constantinescu, P. Lenzner, R. Reiffenhäuser, D. Schmand and G. Varricchio} %TODO mandatory. First: Use abbreviated first/middle names. Second (only in severe cases): Use first author plus 'et al.'

\Copyright{Andrei Constantinescu, Pascal Lenzner, Rebecca Reiffenhäuser, Daniel Schmand and Giovanna Varricchio} %TODO mandatory, please use full first names. LIPIcs license is "CC-BY";  http://creativecommons.org/licenses/by/3.0/


 %TODO mandatory; please add comma-separated list of keywords

\category{Track A: Algorithms, Complexity and Games}

\acknowledgements{
We would like to thank the Schloss Dagstuhl-Leibniz-Zentrum f\"ur Informatik who hosted the event ``Computational Social Dynamics'' (Seminar 22452) in November 2022~\cite{DagstuhlSeminar22452}, where the work leading to this paper was started. We are grateful to Martin Hoefer, Sigal Oren, and Roger Wattenhofer for organizing this event. We additionally thank Roger Wattenhofer for the useful discussions concerning this work. We thank Tamio-Vesa Nakajima for contributing to the hardness proof in \cref{app:hardness}.
}%optional

\nolinenumbers %uncomment to disable line numbering

\EventEditors{Karl Bringmann, Martin Grohe, Gabriele Puppis, and Ola Svensson}
\EventNoEds{4}
\EventLongTitle{51st International Colloquium on Automata, Languages, and Programming (ICALP 2024)}
\EventShortTitle{ICALP 2024}
\EventAcronym{ICALP}
\EventYear{2024}
\EventDate{July 8--12, 2024}
\EventLocation{Tallinn, Estonia}
\EventLogo{}
\SeriesVolume{297}
\ArticleNo{42}
\begin{document}

\maketitle
\begin{abstract}
A decade ago, Gerhard Woeginger posed an open problem that became well-known as ``Woeginger's Hiking Problem'': Consider a group of $n$ people that want to go hiking; everyone expresses preferences over the size of their hiking group in the form of an interval between $1$ and $n$. Is it possible to efficiently assign the $n$ people to a set of hiking subgroups so that every person approves the size of their assigned subgroup? The problem is also known as efficiently deciding if an instance of an anonymous Hedonic Game with interval approval preferences admits a wonderful partition.

We resolve the open problem in the affirmative by presenting an $O(n^5)$ time algorithm for Woeginger's Hiking Problem. Our solution is based on employing a dynamic programming approach for a specific rectangle stabbing problem from computational geo\-metry. Moreover, we propose natural, more demanding extensions of the problem, e.g., maximizing the number of satisfied participants and variants with single-peaked preferences, and show that they are also efficiently solvable.
Last but not least, we employ our solution to efficiently compute a partition that maximizes the egalitarian welfare for anonymous single-peaked Hedonic Games.
\end{abstract}

% intro %
\section{Introduction}
Suppose there are $n$ attendees of a workshop, who aim to go for a joint hike during a break. Keeping the whole group together is logistically challenging, so typically the attendees split into smaller subgroups that will do the hike together. 
It is natural that different attendees might have different preferences for the sizes of the subgroups they will eventually join. In particular, each attendee $i$ reports an interval $[\ell_i, r_i]$ signifying that they will be content if they are in a group of size $s \in [\ell_i, r_i]$. The organizers of the hike now face the following problem: Is there a polynomial time algorithm that determines whether there is a partition of the attendees into subgroups such that they are all content with the sizes of their subgroups?

In 2013, Gerhard Woeginger famously phrased this problem underlying an anonymous hedonic game to explain an open question stemming from his work~\cite{woeginger2013core}. It is one of the very nice problems that he used to share at various occurrences, e.g., at the coffee machine, while waiting in a seminar room, or even during a joint walk.\footnote{Gerhard's ability to explain open research questions in an easily accessible way has been extraordinarily motivating. In particular, this work would not have started without Gerhard meeting one of the authors for lunch and explaining the problem exactly in this way. With this work we contribute to the recent line of publications celebrating the life and work of Gerhard Woeginger, see \cite{in_memoriam_gerhard}.} As this is one of the problems he explained to many people, it became known as Woeginger's \emph{Hiking Problem}.

In this work we answer Woeginger's question in the affirmative by exploiting a tight connection to a variant of the rectangle stabbing problem from computational geometry that can be solved in polynomial time via elegant dynamic programming. 

\begin{figure}[t]%
\begin{center}
\centering%
 \includegraphics[width=.33\linewidth]{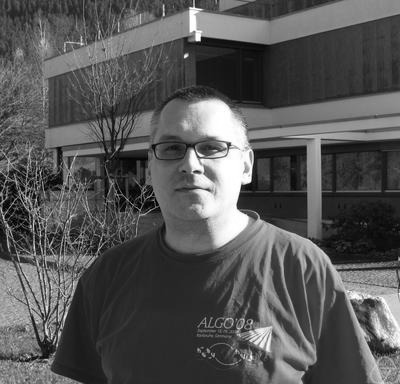}%
 \caption{Gerhard Woeginger, Oberwolfach, 2011. \cite{photo_gerhard}}% 
 \label{fig:gerhard}%
\end{center}
\end{figure}%

Using Woeginger's motivation of the problem, many natural extensions arise that we introduce in this paper.
If the sought partition does not exist, then what is the maximum number of attendees that can be satisfied with their group sizes, or what is the minimum number of attendees to exclude such that the remaining ones can be partitioned to all become satisfied? 
Moreover, we also consider a version where the hikers have single-peaked preferences over the group sizes and a partition is sought that minimizes the utilitarian or egalitarian cost, where cost is defined as a function of the assigned and the ideal group size.
We show that these more demanding problems can also be solved efficiently.
Finally, we discuss the relationship between the hiking problem and the problem of maximizing the egalitarian welfare in (general) anonymous Hedonic Games.

% related %

\subsection{Related Work}

\emph{Hedonic Games} (HGs), introduced by Dreze and Greenberg in~\cite{dreze1980hedonic}, model multi-agent systems where selfish agents have to be partitioned into coalitions and have preferences over the possible outcomes. Such games are called hedonic as agents’ preferences only depend on the coalition they belong to but not on how the other agents are grouped. 
HGs have been widely studied (see~\cite{aziz2016hedonic} for a survey) and numerous prominent subclasses have been identified based on properties of the agents' preferences or other possible constraints:~\cite{ballester2004np,bogomolnaia2002stability,alcalde2004researching,olsen2009nash,Elkind09, aziz2019fractional}. 
Simple examples of such classes are anonymous HGs~\cite{banerjee2001core}, where the preferences of the agents depend only on the size of their coalition and not on the individual participants, or HGs with approval-based (Boolean) preferences~\cite{aziz2016boolean}, where agents have binary values for the coalitions. Woeginger’s Hiking Problem resides in the intersection of these classes.

The HGs literature has typically focused on the existence and computation of stable or optimal solutions, see, e.g., \cite{bogomolnaia2002stability}. The most desirable, ideal partition is the one where every agent is assigned to one of her best coalitions --- called {\em perfect} (or {\em wonderful}) {\em partition}~\cite{aziz2013pareto}. Unfortunately, such a solution rarely exists and the related decision or computational problem is usually hard~\cite{aziz2013pareto,peters2016complexity}.
Even in simple cases such as anonymous approval-based HGs, it is NP-complete to determine the existence of a wonderful partition~\cite{Darmann_journal,woeginger2013core}. 
Such a result holds true even if the number of approved coalition sizes of each agent is at most $2$~\cite{Darmann_journal}.
The problem of finding a wonderful partition in Hedonic Games is related to the one of maximizing the utilitarian welfare, i.e.\ the sum of agents' utilities, for Boolean utilities. An overall picture of the complexity of finding wonderful partitions in Boolean Hedonic Games, including anonymous ones, is given in~\cite{peters2016complexity}. For general utility functions, the problem of maximizing the utilitarian or egalitarian social welfare has been studied in various settings, e.g., in fractional~\cite{AGGMT15} and additively separable Hedonic Games~\cite{ABG16}. To the best of our knowledge, these objectives have not been considered in the context of anonymous Hedonic Games.

In our paper, we show the tractability of computing a wonderful partition for instances with interval approvals as formulated by Woeginger~\cite{woeginger2013core}. To this aim, we provide a dynamic program that relies on the approach used in~\cite{rectangle_stabbing} to solve a capacitated rectangle stabbing problem from computational geometry. 
 Here, the goal is to stab a set of rectangles with a minimum subset of a given set of lines, each intersecting (i.e., potentially \emph{stabbing}) some of the rectangles. Each line has a maximum number of rectangles it can stab, and to stab all rectangles, one is allowed to use multiple copies of each line.

 We also consider variants with single-peaked cost functions of the agents. Single-peaked preferences date back to Black~\cite{black1948rationale}. Such preferences are a common theme in the Economics and Game Theory literature. In particular, they play a prominent role in different fields such as Hedonic Diversity Games~\cite{BredereckEI19,BoehmerE20}, Schelling Games~\cite{BiloBLM22,FriedrichLMS23}, and in various works on voting and social choice~\cite{Walsh07,YCE13,BSU13,ElkindFS20,BBHH15}.

% Preliminaries
\subsection{Model}
The hiking problem as formulated by Gerhard Woeginger is a special case of anonymous Hedonic Games with approval-based preferences.
In an anonymous Hedonic Game with approval-based preferences we are given a set $N$ of $n$ agents to be partitioned into coalitions. Each agent $i \in N$ reports an approval set $S_i \subseteq [n]$ representing the approved group sizes for agent $i$. In particular, agent $i$ wants to be in some group of size $s$ such that $s\in S_i.$ An approval set $S_i$ is said to be an \emph{interval} if $S_i=\set{\ell_i, \ell_i+1, \dots, r_i}$, for some $1 \leq \ell_i \leq r_i \leq n$. 
The agents have to be partitioned into coalitions, i.e., subsets of the agent set. This induces a \emph{partition} $\pi$ of the set of agents $N$. We denote by $\pi(i)$ the coalition agent $i$ belongs to in the partition $\pi$.
We follow the notation in Woeginger's survey paper \cite{woeginger2013core} and call a partition $\pi$ \emph{wonderful} if each agent approves of the size of its coalition in $\pi$, i.e., for each agent $i$, $\modulus{\pi(i)}\in S_i$. This leads to the following natural computational problem called \fauxsc{Wonderful-Partition}.

\begin{tcolorbox}
\textbf{\fauxsc{Wonderful-Partition}}\\
\textbf{Input}: A set $N$ of agents and size approval sets $(S_i)_{i \in N}.$ \\
\textbf{Problem}: Decide whether there exists a wonderful partition of the agents. If yes, compute one.
\end{tcolorbox}

Woeginger's Hiking Problem is \fauxsc{Wonderful-Partition} with interval approval sets, which we formally define as follows.

\begin{tcolorbox}
\textbf{\fauxsc{Wonderful-Partition-Intervals} (\fauxsc{Hiking})}\\
\textbf{Input}: A set $N$ of agents and for each agent $i \in N$ two numbers $\ell_i \leq r_i$ such that $S_i = \set{\ell_i, \ldots, r_i}.$ \\
\textbf{Problem}: Decide whether there exists a wonderful partition of the agents. If yes, compute one.
\end{tcolorbox}

We also consider natural extensions of the hiking problem, to be introduced later.

\subsection{Our Contribution}
We solve Woeginger's Hiking Problem in the affirmative by giving an $O(n^5)$-algorithm that computes a wonderful partition for an instance with $n$ agents that each have interval approval sets (see \Cref{fig:utilities}~(a)), if such a partition exists. For this we use a dynamic programming approach for a rectangle stabbing problem from computational geometry. Moreover, we extend this approach to achieve an $O(n^5)$-algorithm for computing the minimum set of hikers that have to be excluded from the hike, in order for a wonderful partition to become possible. We also give an $O(n^7)$-algorithm for deciding if a wonderful partition exists if exactly $x$ hikers are excluded. This is then used to derive an $O(n^7 \log n)$-algorithm for finding a partition that maximizes the number of hikers that approve of their assigned group size. These approaches can also be extended to a setting where hikers have weights.
\begin{figure}[b]
\includegraphics[width=\textwidth]{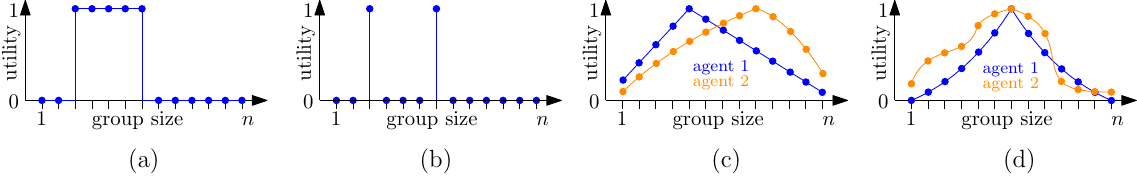}
\caption{Examples of utility functions considered in this paper. (a) interval approval sets, (b) non-interval approval sets with two approved group sizes, (c) utility functions are single-peaked and all agents with the same peak have the same utility function (under some additional mild technical assumptions), (d) individual single-peaked utility functions for all agents (as shown, the functions may differ even if they have the same peak).}
\label{fig:utilities}
\end{figure}

All these positive results hold for the case where the agents have interval approval sets (\Cref{fig:utilities}~(a)). The special case where every hiker only approves of one group size, i.e., approval intervals of size $1$, 
can be solved efficiently by checking if for all $i$ the number of agents approving only of size $i$ is divisible by $i.$
In contrast to this, it was already known that the problem is NP-hard even if each attendee has at most two different approved group sizes that need not form an interval~\cite{Darmann_journal} (see \Cref{fig:utilities}~(b)). One can modify the original proof to establish hardness even for the case where each attendee approves of \emph{exactly} two group sizes. For the sake of completeness and to improve readability, we show the details in \cref{app:hardness}. Our presentation hinges on an elegant connection to a graph orientation problem that has not been used in the previous proof.

We complete the picture by considering agents with single-peaked preferences over group sizes. In this setting an agent has a cost that is a function of its assigned group size and its ideal group size. If all agents incur a cost that is given by a fixed function that depends on their peak (but different peak values are allowed, see \Cref{fig:utilities}~(c)) and if this function satisfies a mild technical condition, we compute a partition that minimizes the social cost in time $O(n^2(\alpha+1))$, if at most $\alpha$ agents can be excluded from the hike. This holds for the utilitarian setting, i.e., minimizing the sum of the agents' costs, and also for the egalitarian setting, where the maximum agent cost is to be minimized. Finally, we prove polynomial time equivalence of the wonderful partition problem and the problem of minimizing the egalitarian social cost in (general) anonymous Hedonic Games. We use this result to show that even for the setting where each agent has its own single-peaked cost function (see \Cref{fig:utilities}~(d)), a partition that minimizes the egalitarian social cost can be computed in $O(n^5\log n)$. 

% All omitted details can be found in the appendix.

% DP %
\section{Efficient Algorithm for Woeginger's Hiking Problem, and Extensions}\label{sec:hiking}
We prove that \wonderfulIntervals\ can be solved in polynomial time by casting it as a version of capacitated rectangle stabbing that is efficiently solvable via an elegant dynamic programming approach introduced by Even, Levi, Rawitz, Schieber, Shahar, and Srividenko~\cite{rectangle_stabbing}. In the capacitated rectangle stabbing problem, we are given axis-parallel rectangles and a set of axis-parallel lines. The goal is to find a minimum number of lines that intersect and `stab' all rectangles, or, in a version where lines have costs, lines with minimum total cost with the same property. Each line has a maximum number of intersected rectangles that it can `stab', and, to stab all rectangles, one is allowed to use multiple copies of each line. Woeginger's Hiking Problem can be seen as a special case of the one-dimensional capacitated rectangle stabbing problem with costs, where the agents' intervals correspond to one-dimensional rectangles and stabbing rectangles at some integer position $i$ costs $i$ and corresponds to creating a group of size $i$ in which we include the agents corresponding to the stabbed rectangles. Solutions with total cost $n$ (which is also a lower bound on the cost) then correspond to those where all groups are full.
Building on the approach of \cite{rectangle_stabbing} applied to this case, we give a simpler $\mathcal{O}(n^5)$ dynamic program that is specifically tailored to our needs. 
Note that this reduction from our partition problem to a version of rectangle stabbing draws a powerful new connection between the problem types that is not limited to one specific variant of the partition problem, but also holds more generally; e.g., it can be adapted to the variant where there can be at most one group of each size.

In the following, we use the central observations of Even, Levi, Rawitz, Schieber, Shahar, and Srividenko~\cite{rectangle_stabbing} to derive a dynamic programming solution to Woeginger's Hiking problem, i.e., to \wonderfulIntervals. Its definition and analysis will then also power our results on the more demanding problem variants discussed in the introduction. 

\subsection{Dynamic Program for Woeginger's Hiking Problem}\label{sec:polyDP}
Consider an instance of \fauxsc{Wonderful-Partition-Intervals}, i.e., a set $N$ of $n$ agents, without loss of generality $N = [n]$, and $n$ pairs $(\ell_i, r_i)_{i \in N}$ such that agent $i \in [n]$ approves of sizes $S_i = \set{\ell_i, \ldots, r_i}$. We are interested in checking whether there exists a wonderful partition of the agents; i.e., one where every agent approves of their coalition size. For consistent tie-breaking reasons, throughout this section we write $i \prec j$ for two distinct agents $i, j \in [n]$ if either $r_i < r_j,$ or $r_i = r_j$ and $i < j.$ Note that $\prec$ is a well-defined strict linear order and, without loss of generality, assume that agents are ordered so that $1 \prec \ldots \prec n.$

As the main technical ingredient of our approach, given a subset $N' \subseteq N$ of agents, we say that a wonderful partition $\pi$ of $N'$ is \emph{earliest-due-date} if for any two agents $i \prec j$ it does not hold that $ \ell_i \leq \modulus{\pi(j)} < \modulus{\pi(i)}.$ This terminology is borrowed from the scheduling literature, and based on the following well-known observation:~let us view the agents' allowed intervals of coalition sizes on an axis labeled with the natural numbers $1,\ldots,n$, and consider an arbitrary collection of coalition sizes we might decide on using, in non-decreasing order of coalition size. Then, going in this order, it is always safe to include those agents first whose right interval endpoints are the smallest: in said order of coalitions, they are the ones that stop being servable first, making an earliest-due-date approach impose the least restrictions on later assignments. In other words, whenever a wonderful partition into the given coalition sizes exists, there also exists one that is earliest-due-date. We prove this formally below for completeness. Note that, to make the process well-defined, instead of comparing agents by right endpoints, we compare them by $\prec$. 

\begin{restatable}{lemma}{lemmaone} If a subset $N' \subseteq N$ of agents admits a wonderful partition, then it admits an earliest-due-date wonderful partition. 
\end{restatable}
\begin{proof} Consider a wonderful partition $\pi$ of $N'.$ If it satisfies the required property, then we are done, otherwise, consider $i, j \in N'$ such that $i\prec j$ and $\ell_i \leq \modulus{\pi(j)} < \modulus{\pi(i)}.$ These conditions imply that $\modulus{\pi(j)} \in [\ell_i, r_i]$ and $\modulus{\pi(i)} \in [\ell_j, r_j].$ Hence, one can construct a new wonderful partition $\pi'$ from $\pi$ by exchanging the groups of agents $i$ and $j$; i.e., $\pi'(i) = (\pi(j) \setminus \set{j}) \cup \set{i}$ and $\pi'(j) = (\pi(i) \setminus \set{i}) \cup \set{j}.$ Subsequently, set $\pi \gets \pi'$ and repeat the same procedure until the required condition is satisfied. To complete the proof, we need to show that this is eventually the case. To do so, since $1 \prec \ldots \prec n,$ each exchange strictly increases the sequence $(\modulus{\pi(n)}, \ldots,\modulus{\pi(1)})$ lexicographically. Because this sequence is bounded from above by $(n, \ldots, n),$ the process eventually ends.
\end{proof}

\noindent Hence, from now on, we will only seek earliest-due-date wonderful partitions. This crucial restriction, which we have shown to be without loss of generality, will allow us to bootstrap a dynamic programming algorithm that determines a wonderful partition if one exists.

To construct our algorithm, we first show that earliest-due-date partitions admit an attractive recursive decomposition. 
Consider an earliest-due-date wonderful partition $\pi$ of the agents (if any exist), and consider the size of the coalition that agent $n$ is a part of; i.e., $\modulus{\pi(n)}.$ Moreover, consider an arbitrary agent $i \neq n.$ Agent $i$ is part of a coalition of size $\modulus{\pi(i)}.$ Because $\pi$ is earliest-due-date, by definition it can not be the case that $\ell_i \leq \modulus{\pi(n)} < \modulus{\pi(i)},$ so either $\modulus{\pi(n)} < \ell_i$ or $\modulus{\pi(i)} \leq \modulus{\pi(n)}$ holds.

\begin{restatable}{lemma}{lemRecursiveDecomp} \label{lemma:recursive-decomposition} Consider an earliest-due-date wonderful partition $\pi$ of $N.$ Partition the agents as $N = N_{-} \cup N_{+} \cup \{n\},$ where $N_{-} = \set{i \in N \setminus \{n\} \mid \ell_i \leq \modulus{\pi(n)}}$ and $N_{+} = \set{i \in N \setminus \{n\} \mid \ell_i > \modulus{\pi(n)}},$ then it holds that:

\begin{enumerate}
    \item  \label{item_1} For all $i \in N_{-}$ we have $\modulus{\pi(i)} \leq \modulus{\pi(n)}$;
    \item \label{item_2} For all $i \in N_{+}$ we have $\modulus{\pi(i)} > \modulus{\pi(n)}.$
\end{enumerate}%
\end{restatable}
\begin{proof} For (\ref{item_2}) note that any $i \in N_{+}$ by definition satisfies $\ell_i > \modulus{\pi(n)}.$ Since we have the requirement that $\modulus{\pi(i)} \in [\ell_i, r_i],$ this means that $\modulus{\pi(i)} \geq \ell_i > \modulus{\pi(n)},$ as desired.

For (\ref{item_1}), consider some $i \in N_{-}.$ By definition, we have that $\ell_i \leq \modulus{\pi(n)}.$ Assume for a contradiction that $\modulus{\pi(i)} > \modulus{\pi(n)}.$ This implies that $\ell_i \leq \modulus{\pi(n)} < \modulus{\pi(i)}.$ Since $r_i \leq r_n$ and by the sorting criterion we also have $i \prec n$, which contradicts that $\pi$ is earliest-due-date.
\end{proof}

\noindent Notably, the same result holds true if we work with a subset of the agents $N' \subsetneq N$ and $n$ is replaced with the agent $a \in N'$ with maximum $r_a.$

Armed as such, to build intuition, Lemma \ref{lemma:recursive-decomposition} tells us that no two agents $a_{-} \in N_{-}$ and $a_{+} \in N_{+}$ can go into the same group, so a first attempt at a recursive algorithm looking for a wonderful partition of $N,$  to be later optimized by memorization/dynamic programming, would proceed as follows: start with $N' = N$; at each step, identify the agent $a \in N'$ maximizing $r_a$ and exhaust over all possibilities for $\modulus{\pi(a)}$; for each one, write $N' = N'_{-} \cup N'_{+} \cup \{a\}$ as defined in Lemma \ref{lemma:recursive-decomposition} and recurse with $N'_{-}$ and $N'_{+}$. We will specify the details of the process such that, if the recursive calls yield wonderful partitions of $N'_{-}$ and $N'_{+}$, a wonderful partition of $N$ can also be constructed by incorporating agent $a$ into them.
Of course, the last step of reasoning is incorrect, since one needs to make sure that a group of size $\modulus{\pi(a)}$ with one space available indeed exists, and this information needs to be somehow propagated across recursive calls. Moreover, it is unclear what the number of sets $N'$ reachable by the recursion is. This has to be polynomially bounded so that memorization/dynamic programming leads to a polynomial-time algorithm.

Let us first address the second issue outlined above, namely the size of the state space. This turns out to be relatively simple:~define $N(x_1, x_2, i) = \set{j \in N \mid j \leq i \text{ and } \ell_j \in [x_1, x_2]}.$ Then, we can adapt the recursive approach:~instead of storing $N'$ explicitly, start with $x_1 = 1$ and $x_2 = n$, as well as $i = n,$ from which implicitly $N' = N(x_1, x_2, i)$. At each step we will first check whether $i$, which by the sorting criterion is the largest agent maximizing $r_i$, is in $N' = N(x_1, x_2, i)$; i.e., whether $\ell_i \in [x_1, x_2].$ If not, then $N(x_1, x_2, i)=N(x_1,x_2,i-1)$ and we simply recurse with the same $x_1, x_2$ and $i' = i - 1.$ Otherwise, $i$ is the largest agent in $N'$ maximizing $r_i$, so we can exhaust as before over all possible values for $\modulus{\pi(i)}$ and perform two recursive calls in each case:~one with $(x_1', x_2', i') = (x_1, \modulus{\pi(i)}, i - 1)$ and one with $(x_1', x_2', i') = (\modulus{\pi(i)} + 1, x_2, i - 1)$. We have yet to solve the correctness issue, but we have made progress:~the state space now consists of triples $(x_1, x_2, i)$ such that $1 \leq x_1 \leq x_2 \leq n$ and $1 \leq i \leq n,$ which there are $O(n^3)$ of. 

Let us now turn our attention to ensuring correctness. To do so, we need to investigate more closely how to ensure that the group of size $\modulus{\pi(i)}$ of agent $i$ (recall that we only exhaust over its size, not over which agents are in it) exists and is used to its full capacity in the solution constructed by the recursion. There is one crucial guarantee given by Lemma \ref{lemma:recursive-decomposition} that we have so far not exploited. Namely, the recursive call with $(x_1, \modulus{\pi(i)}, i - 1)$ should only use groups of sizes $s \in [x_1, \modulus{\pi(i)}]$ and the one with $(x_1', x_2', i') = (\modulus{\pi(i)} + 1, x_2, i - 1)$ should only use groups of sizes $s \in [\modulus{\pi(i)} + 1, x_2].$ In general, state $(x_1, x_2, i)$ should only consider groups of sizes $s \in [x_1, x_2].$ This can be ensured by requiring that the exhausted value for $\modulus{\pi(i)}$ additionally satisfies $x_1 \leq \modulus{\pi(i)} \leq x_2.$ This small, seemingly inconsequential, refinement of the approach will ultimately allow us to propagate information about incomplete groups across states in the recursion. Before describing how this is done in general, to build more intuition, let us consider the first level of the recursion, namely $(x_1, x_2, i) = (1, n, n).$ At this level, the algorithm exhausts over the possible values for $\modulus{\pi(n)} \in [\ell_n, r_n].$ For each possibility, in the ensuing recursive call with $(x_1', x_2', i') = (\modulus{\pi(n)} + 1, n, n - 1)$ all used groups will be of size at least $\modulus{\pi(n)} + 1 > \modulus{\pi(n)}$ so nothing special needs to be done in this case since those agents can not be part of agent $n$'s group. However, in the other call, having $(x_1', x_2', i') = (1, \modulus{\pi(n)}, n - 1)$ some $\modulus{\pi(n)} - 1$ other agents will need to share a group of size $\modulus{\pi(n)}$ with agent $n,$ and this is not yet modeled by our approach. To model this effect across recursive calls, we introduce a fourth variable $0 \leq k < x_2$ to the state of our recursive algorithm; i.e., each state is now of the form $(x_1, x_2, i, k).$ This variable intuitively signifies that there exists (from upward in the recursion) an incomplete group, currently of size $k,$ whose final size should be $x_2.$ With this setup, the starting top-level call will now be with $(x_1, x_2, i, k) = (1, n, n, 0).$ For each value of $\modulus{\pi(n)},$ the two resultant recursive calls will be with $(x_1', x_2', i', k') = (1, \modulus{\pi(n)}, n - 1, 1)$\footnote{Strictly speaking, this should be $(x_1', x_2', i', k') = (1, \modulus{\pi(n)}, n - 1, 1 \Mod {\modulus{\pi(n)}})$. We omit this detail from the higher-level exposition to improve readability.} and $(x_1', x_2', i', k') = (\modulus{\pi(n)} + 1, n, n - 1, 0).$ Note how the former call has $k' = 1$,  signifying that we just ``opened a new (incomplete) group of size one, whose final size should be $x_2' = \modulus{\pi(n)}$''
The fundamental reason why such an approach can work is that in any node of the recursion tree, there is at most a single incomplete group to keep track of. Note that this fact crucially depends on each call $(x_1, x_2, i, k)$ only considering partitions into groups of sizes between $x_1$ and $x_2.$ We still need to describe the transitions for general calls $(x_1, x_2, i, k)$ given the newly added parameter~$k$. The formal details will following below, but in rough lines, the call for $N_{-}$ creates a new group; i.e., $k' = 1$; while the call for $N_{+}$ keeps the currently open one; i.e., $k' = k$; the exception comes when $\modulus{\pi(i)} = x_2,$ in which case the call for $N_{-}$ adds to the same group, possibly closing the group; i.e., $k' = (k + 1) \pmod{x_2}$; and no call for $N_{+}$ is generated.

\begin{theorem}  \label{thm:hiking-is-poly-time} Deciding whether a wonderful partition exists and computing one if so can be achieved in $O(n^5)$ time.
\end{theorem}
\begin{proof} We show how to solve the decision version. Constructing a wonderful partition for yes-instances can be subsequently achieved by standard techniques. We proceed by dynamic programming (DP). Define the Boolean DP array $\mathit{dp}[x_1, x_2, i, k]$ for $1 \leq x_1 \leq x_2 \leq n,$ $0 \leq i \leq n$ and $0 \leq k < x_2,$ with the meaning $\mathit{dp}[x_1, x_2, i, k] = 1$ if and only if there exists a wonderful partition of agents in $N(x_1, x_2, i)$ using only groups of sizes in $[x_1, x_2]$ and assuming that we start with an incomplete group of current size $k$ which has to have final size $x_2.$ Naturally, if $k = 0$ this means starting with no partially filled group. The final answer will be available at the end in $\mathit{dp}[1, n, n, 0].$ To compute the DP table, we use the following recurrence relations and base cases (using a hyphen as stand-in for any group size):

\begin{enumerate}
    \itemsep0em 
    \item $\mathit{dp}[-, -, 0, 0] = 1$;
    \item $\mathit{dp}[-, -, 0, k] = 0$ if $k \neq 0$;
    \item $\mathit{dp}[x_1, x_2, i, k] = \mathit{dp}[x_1, x_2, i - 1, k]$ if $\ell_i \notin [x_1, x_2]$;
    \item \label{state_type_4} $\mathit{dp}[x_1, x_2, i, k] = \bigvee_{x = \ell_i}^{\min(x_2, r_i)}F_x$ if $\ell_i \in [x_1, x_2],$ where 
    \[ 
F_x := \left\{
\begin{array}{ll}
      \mathit{dp}[x_1, x_2, i - 1, (k + 1) \Mod{x_2}] & x = x_2 \\
      \mathit{dp}[x_1, x, i - 1, 1 \Mod{x}] \land \mathit{dp}[x + 1, x_2, i - 1, k] & x < x_2. \\
\end{array} 
\right. 
\]
\end{enumerate}
The first three cases are immediate from the definition of the DP. For the last case, we iterate over $x \in [\ell_i, \max(x_2, r_i)]$ which is agent $i$'s group size. There are two cases:~if $x = x_2,$ then we put $i$ into the group of current size $k$ and final size $x_2$ that we assumed to have at our disposal; this increases the size of the group by one, or completes it if $k = x_2 - 1$; otherwise, $x < x_2,$ and we recurse into partitioning agents in $N(x_1, x, i - 1)$ into groups of sizes between $x_1$ and $x,$ and $N(x + 1, x_2, i - 1)$ into groups of sizes between $x + 1$ and $x_2.$ For the first call, this generates a new group of size $1$ (unless $x = 1$), while for the latter, this keeps the previously open group of current size $k$ and final size $x_2$ that we assumed to have.

To compute the DP table in an acyclic fashion, it suffices to iterate through $i$ in ascending order. The complexity of the approach is $O(n^5)$ because there are $O(n^4)$ states and computing the value for states of type~(\ref{state_type_4}) requires iterating through $O(n)$ values of $x.$
\end{proof}

\subsection{Extensions}

Using a similar DP approach, we can solve the following natural extension.

\begin{tcolorbox}
\textbf{\fauxsc{Hiking-Min-Delete}}\\
\textbf{Input}: A set $N$ of agents and for each agent $i \in N$ two numbers $\ell_i \leq r_i$ such that $S_i = \set{\ell_i, \ldots, r_i}.$ \\
\textbf{Problem}: Compute a set $N' \subseteq N$ of minimum size such that $N \setminus N'$ has a wonderful partition. Output $N'$ and a wonderful partition of $N \setminus N'.$
\end{tcolorbox}

\noindent The approach relies on essentially the same recursive reasoning as before, except that we now also consider the possibility of ``ignoring'' an agent $i$ and hence recursing with $(x_1', x_2', i', k') = (x_1, x_2, i - 1, k).$ Moreover, instead of making the recursion return whether a wonderful partition is possible or not, we make it return the minimum number of agents that need to be removed so that this is possible. In light of this, the ``ignore $i$'' recursive call incurs a cost of 1 removed agent. The details are formalized in the following.

\begin{theorem}\label{th:hiking-min-delete-poly-time} \fauxsc{Hiking-Min-Delete} is solvable in $O(n^5)$ time.
\end{theorem}
\begin{proof} We use a similar approach as in the proof of Theorem \ref{thm:hiking-is-poly-time}. As before, it suffices to show how to compute the minimum size of $N',$ as computing such a set $N'$ and a corresponding wonderful partition for $N \setminus N'$ follow by standard techniques. Define the integer-valued DP array $\mathit{dp}[x_1, x_2, i, k]$ for $1 \leq x_1 \leq x_2 \leq n,$ $0 \leq i \leq n$ and $0 \leq k < x_2,$ with the meaning that $\mathit{dp}[x_1, x_2, i, k]$ contains the minimum number of agents which need to be removed from $N(x_1, x_2, i)$ so that the remaining agents admit a wonderful partition into groups of sizes in $[x_1, x_2]$ assuming that we start with an incomplete group of current size $k$ which has to have final size $x_2.$ Note that, in contrast to the previous DP, the value 0 corresponds to removing no agents, which is the best possible outcome, while previously it corresponded to the need for removing at least one agent.
The final answer will be available at the end in $\mathit{dp}[1, n, n, 0].$ To compute the DP table, we use the following recurrence relations and base cases:

\begin{enumerate}
    \item $\mathit{dp}[-, -, 0, 0] = 0$;
    \item $\mathit{dp}[-, -, 0, k] = \infty$ if $k \neq 0$;
    \item $\mathit{dp}[x_1, x_2, i, k] = \mathit{dp}[x_1, x_2, i - 1, k]$ if $\ell_i \notin [x_1, x_2]$;
    \item $\mathit{dp}[x_1, x_2, i, k] = \min\{1 + \mathit{dp}[x_1, x_2, i - 1, k], X\}$ if $\ell_i \in [x_1, x_2]$, where \[X := \min\{F_x \mid \ell_i \leq x \leq \min(x_2, r_i)\}\] and  
    \[ 
F_x := \left\{
\begin{array}{ll}
      \mathit{dp}[x_1, x_2, i - 1, (k + 1) \Mod{x_2}] & x = x_2 \\
      \mathit{dp}[x_1, x, i - 1, 1 \Mod{x}] + \mathit{dp}[x + 1, x_2, i - 1, k] & x < x_2. \\
\end{array} 
\right. 
\]
\end{enumerate}

\noindent The reasoning stays largely the same as before, with the only difference being the new $1 + \mathit{dp}[x_1, x_2, i - 1, k]$ term, which accounts for discarding agent $i$ and incurring a cost of 1 corresponding to removing an agent. 
\end{proof}

\noindent Another subtly distinct natural variant of Woeginger's Hiking Problem is the following:

\begin{tcolorbox}
\textbf{\fauxsc{Hiking-Max-Satisfied}}\\
\textbf{Input}: A set $N$ of agents and for each agent $i \in N$ two numbers $\ell_i \leq r_i$ such that $S_i = \set{\ell_i, \ldots, r_i}.$ \\
\textbf{Problem}: Compute a partition $\pi$ of $N$ maximizing the number of agents approving of their coalition sizes.
\end{tcolorbox}

\noindent Indeed, \fauxsc{Hiking-Max-Satisfied} and \fauxsc{Hiking-Min-Delete} are related, in that if $N'$ is a minimum-size set of agents such that $N \setminus N'$ has a wonderful partition, then it is also possible to satisfy at least $n - \modulus{N'}$ agents in a partition of $N.$ This is because agents in $N'$ can be put together in a group in tandem to a wonderful partition of $N \setminus N'$ to get a partition of $N$ with at least $n - \modulus{N'}$ satisfied agents. However, it might be possible to satisfy more than $n - \modulus{N'}$ agents if unsatisfied agents do not all go into the same group.

Hence, solving this variant of the problem introduces new challenges that require further insight. Let us fix a number $k$ and ask whether it is possible to satisfy at least $\modulus{N} - k$ agents. If the answer is affirmative, we also want a partition achieving this. To solve \fauxsc{Hiking-Max-Satisfied}, we will binary search for the smallest $0 \leq k \leq \modulus{N}$ for which the answer is affirmative, call it $k^*$, and then recover a partition for $k^*.$ It remains to show how to solve the problem for a fixed value of $k.$  To do so, we need to adjust the angle from which we look at the problem. In particular, instead of looking for a partition satisfying at least $k$ agents, we will look for a size-$k$ subset $N' \subseteq N$ (corresponding to $k$ agents which we do not require to be satisfied) such that $(N \setminus N') \cup D_{k}$ admits a wonderful partition, where $D_{k}$ is a set of $k$ dummy agents happy with any group size. Intuitively, $k$ agents are replaced with dummies not minding their group size. Because it is always no worse to remove agents from $N$ in contrast to removing agents from $D_k,$ it is enough to ask to remove exactly $k$ agents from $N \cup D_k$ such that the remaining agents admit a wonderful partition. Checking whether this is possible and finding a corresponding wonderful partition reduces to the following more general problem, which we will show can be solved in polynomial time.

\begin{tcolorbox}
\textbf{\fauxsc{Hiking-x-Delete}}\\
\textbf{Input}: A set $N$ of agents, for each agent $i \in N$ two numbers $\ell_i \leq r_i$ such that $S_i = \set{\ell_i, \ldots, r_i},$ and also a number $0 \leq x \leq \modulus{N}.$ \\
\textbf{Problem}: Compute a set $N' \subseteq N$ of size $x$ such that $N \setminus N'$ has a wonderful partition (or report impossibility). Output $N'$ and a wonderful partition of $N \setminus N'.$
\end{tcolorbox}

We now show how to solve \fauxsc{Hiking-x-Delete} in polynomial time. We can once again try to attack the problem recursively with our usual state $(x_1, x_2, i, k).$ Just like we did for \fauxsc{Hiking-Min-Delete}, we will have recursive calls for ignoring an agent; i.e., adding it to the set $N'$ of removed agents.\footnote{Note that previously we used $N'$ to denote $N(x_1, x_2, i)$ when discussing the recursive algorithm for the original hiking problem. We do not keep this notation here.} In order to ensure that exactly $x$ agents are removed, we add another variable $r$ to the state of the recursion: $(x_1, x_2, i, k, r),$ where $r$ is how many agents we want to remove. The top-level call will be invoked with $r = x.$ The recursive approach for \fauxsc{Hiking-Min-Delete} now translates relatively swiftly to the new setting. The main difference is the case where a state of the form $(x_1, x_2, i, k, -)$ with $\ell_i \in [x_1, x_2]$ performs two recursive calls to $(x_1, x, i - 1, 1, -)$ and $(x + 1, x_2, i - 1, k, -).$ In particular, say the state is $(x_1, x_2, i, k, r),$ then what should be the values $r'$ and $r''$ for the two recursive calls? Intuitively, there is no fixed answer, since it could be that we remove more agents in the first or in the second call. In fact, we have full freedom over how to split the $r$ removals across the two calls as long as $r' + r'' = r.$ Hence, we will iterate over all options $0 \leq r' \leq r$ and set $r'' = r - r',$ and in each case call recursively with $(x_1, x, i - 1, 1, r')$ and $(x + 1, x_2, i - 1, k, r'').$

\begin{theorem}\label{lemma:hiking-x-delete-poly} \fauxsc{Hiking-x-Delete} is solvable in $O(n^7)$ time.
\end{theorem}
\begin{proof} As before, it suffices to check feasibility, a solution can then be recovered using standard techniques. Define the Boolean DP array $\mathit{dp}[x_1, x_2, i, k, r]$ for $1 \leq x_1 \leq x_2 \leq n,$ $0 \leq i \leq n,$ $0 \leq k < x_2,$ and $0 \leq r \leq x$ with the meaning that $\mathit{dp}[x_1, x_2, i, k, r] = 1$ if and only if there exist $r$ agents that can be removed from $N(x_1, x_2, i)$ such that the remaining agents admit a wonderful partition into groups of sizes in $[x_1, x_2]$ assuming that we start with an incomplete group of current size $k$ which has to have final size $x_2.$ At the end, $\mathit{dp}[1, n, n, 0, x]$ will be 1 if and only if it is possible to remove $x$ agents from $N$ such that the remaining agents admit a wonderful partition. To compute the DP table, we use the following:

\begin{enumerate}
    \item \label{dp3_state_type_1} $\mathit{dp}[-, -, 0, 0, 0] = 1$;
    \item \label{dp3_state_type_2} $\mathit{dp}[-, -, 0, k, r] = 0$ if $k \neq 0$ or $r \neq 0$;
    \item \label{dp3_state_type_3} $\mathit{dp}[x_1, x_2, i, k, r] = \mathit{dp}[x_1, x_2, i - 1, k, r]$ if $\ell_i \notin [x_1, x_2]$;
    \item \label{dp3_state_type_4} $\mathit{dp}[x_1, x_2, i, k, r] = \mathit{dp}[x_1, x_2, i - 1, k, r - 1, d] \lor X$\footnote{Only $X$ for $r = 0$ as otherwise we would be referring to the invalid value $r = -1.$} if $\ell_i \in [x_1, x_2]$, where $X := \bigvee_{x = \ell_i}^{\min(x_2, r_i)}F_x$ and  
    \[ 
F_x := \left\{
\begin{array}{ll}
      \mathit{dp}[x_1, x_2, i - 1, (k + 1) \Mod{x_2}, r] & x = x_2 \\
      \bigvee_{r' = 0}^{r}\big(\mathit{dp}[x_1, x, i - 1, 1 \Mod{x}, r']
      \land \mathit{dp}[x + 1, x_2, i - 1, k, r - r']\big) & x < x_2. \\
\end{array} 
\right. 
\]
\end{enumerate}

We explain the four cases separately: for cases~(\ref{dp3_state_type_1}) and (\ref{dp3_state_type_2}) we have $i = 0$; i.e., the set of yet-to-be-considered agents is empty, so $k = r = 0$ is necessary and sufficient for the table to store a 1, signifying feasibility. As before, case~(\ref{dp3_state_type_3}) straightforwardly ignores an agent that is not part of the current set of active agents.

The more interesting case is case~(\ref{dp3_state_type_4}). First, the $\mathit{dp}[x_1, x_2, i - 1, k, r - 1]$ term corresponds to removing an agent, and it only applies to $r \geq 1$, as in the footnote. This decreases $r$ by 1 since we have just removed an agent. The $X := \bigvee_{x = \ell_i}^{\min(x_2, r_i)}F_x$ term is more involved. The selection of $\ell_i \leq x \leq \min(x_2, r_i)$ corresponds to selecting the size of the group that agent $i$ will be part of. For $x = x_2,$ the analysis is similar to our previous DPs. For $x < x_2,$ on the other hand, we moreover split the $r$ agent removals into $r'$ removals for the first recursive call and $r'' = r - r'$ for the second. Whether both calls are successful is represented by the expression $\mathit{dp}[x_1, x, i - 1, 1 \Mod{x}, r'] \land \mathit{dp}[x + 1, x_2, i - 1, k, r''].$

As before, to compute the DP table in an acyclic fashion, it suffices to iterate through $i$ in ascending order. The complexity is $O(n^7)$ because there are $O(n^5)$ states and computing the value for states of type (\ref{dp3_state_type_4}) requires iterating through $O(n^2)$ values for $(x, r).$    
\end{proof}

\noindent Using the above binary search approach we get a solution for \fauxsc{Hiking-Max-Satisfied} that runs only a $O(\log n)$ factor slower than the runtime for \fauxsc{Hiking-x-Delete}.

\begin{theorem}\label{th:hiking-max-satisfied-poly} \fauxsc{Hiking-Max-Satisfied} is solvable in $O(n^7 \log n)$ time.
\end{theorem}

\subsection{Further Weighted Extensions}
If not all agents can be satisfied,
\fauxsc{Hiking-Min-Delete} and \fauxsc{Hiking-Max-Satisfied} provide two ways of implementing a compromise. However, both treat unsatisfied/deleted agents equally. In certain settings, it might be more desirable to take into account the different \emph{entitlements} of the agents; i.e., one agent might have been dissatisfied with their group size during the previous edition of the workshop, or another agent might be the senior invited speaker. One modelling option is to assign a weight $w_i$ to each agent $i \in N$ and weigh the dissatisfied/deleted agents accordingly leading to the following variants:
\begin{tcolorbox}
\textbf{\fauxsc{Hiking-Min-Delete-Weighted}}\\
\textbf{Input}: A set $N$ of agents and for each agent $i \in N$ two numbers $\ell_i \leq r_i$ such that $S_i = \set{\ell_i, \ldots, r_i}.$ Moreover, for each agent $i \in N,$ a number $w_i \in \mathbb{R}_{\geq 0}.$ \\
\textbf{Problem}: Compute a set $N' \subseteq N$ minimizing $\sum_{i \in N'}w_i$ such that $N \setminus N'$ has a wonderful partition. Output $N'$ and a wonderful partition of $N \setminus N'.$
\end{tcolorbox}

\begin{tcolorbox}
\textbf{\fauxsc{Hiking-Max-Satisfied-Weighted}}\\
\textbf{Input}: A set $N$ of agents and for each agent $i \in N$ two numbers $\ell_i \leq r_i$ such that $S_i = \set{\ell_i, \ldots, r_i}.$ Moreover, for each agent $i \in N,$ a number $w_i \in \mathbb{R}_{\geq 0}.$ \\
\textbf{Problem}: Compute a partition $\pi$ of the agents such that if $N_\pi$ is the set of agents approving of their coalition sizes in $\pi,$ then $\sum_{i \in N_\pi}w_i$ is maximized.
\end{tcolorbox}

Our dynamic programs, with minor modifications which we sketch next, can also be used to solve the weighted variants. We begin with \fauxsc{Hiking-Min-Delete-Weighted}.

\begin{theorem} \fauxsc{Hiking-Min-Delete-Weighted} is solvable in $O(n^5)$ time.
\end{theorem}
\begin{proof}[Proof Sketch] In the proof of \cref{th:hiking-min-delete-poly-time}, we defined the DP as follows: ``$\mathit{dp}[x_1, x_2, i, k]$ contains the minimum number of agents which need to be removed from $N(x_1, x_2, i)$ so that the remaining agents admit a wonderful partition into groups of sizes in $[x_1, x_2]$ assuming that we start with an incomplete group of current size $k$ which has to have final size $x_2$''. To handle weights, this is replaced by ``$\mathit{dp}[x_1, x_2, i, k]$ contains the minimum total weight of agents which need to be removed from $N(x_1, x_2, i)$ so that the remaining agents admit a wonderful partition into groups of sizes in $[x_1, x_2]$ assuming that we start with an incomplete group of current size $k$ which has to have final size $x_2$''. The rest of the proof stays the same, with the minor adaptation that the fourth recurrence relation accordingly becomes $\mathit{dp}[x_1, x_2, i, k] = \min\{w_i + \mathit{dp}[x_1, x_2, i - 1, k], X\},$ where previously $w_i = 1$ has been used.
\end{proof}

\noindent To get a similar result for \fauxsc{Hiking-Max-Satisfied-Weighted} following our previous proof outline, we first define a weighted analogue of \fauxsc{Hiking-x-Delete}, as follows.

\begin{tcolorbox}
\textbf{\fauxsc{Hiking-x-Delete-Min-Weight}}\\
\textbf{Input}: A set $N$ of agents, for each agent $i \in N$ two numbers $\ell_i \leq r_i$ such that $S_i = \set{\ell_i, \ldots, r_i},$ and also a number $0 \leq x \leq \modulus{N}.$ Moreover, for each agent $i \in N,$ a number $w_i \in \mathbb{R}_{\geq 0}.$ \\
\textbf{Problem}: Compute a set $N' \subseteq N$ of size $x$ minimizing $\sum_{i \in N'}w_i$ such that $N \setminus N'$ has a wonderful partition (or report impossibility). Output $N'$ and a wonderful partition of $N \setminus N'.$
\end{tcolorbox}

\begin{restatable}{theorem}{xdeleteweight}
\label{lemma:hiking-x-delete-min-weight-poly} \fauxsc{Hiking-x-Delete-Min-Weight} is solvable in $O(n^7)$ time.
\end{restatable}
\begin{proof}[Proof Sketch] In the proof of \cref{lemma:hiking-x-delete-poly} for the unweighted version, we defined a boolean DP as follows: ``$\mathit{dp}[x_1, x_2, i, k, r] = 1$ if and only if there exist $r$ agents that can be removed from $N(x_1, x_2, i)$ such that the remaining agents admit a wonderful partition into groups of sizes in $[x_1, x_2]$ assuming that we start with an incomplete group of current size $k$ which has to have final size $x_2$''. For the current problem, we want the states to signal not only possibility/impossibility but also what is the minimum total weight of those $r$ removed agents. Hence, we replace this definition by ``$\mathit{dp}[x_1, x_2, i, k, r]$ is the minimum total weight of $r$ agents that can be removed from $N(x_1, x_2, i)$ such that [...] (or $\infty$ if impossible)''. The rest of the reasoning stays analogous with minor changes: the values $0, 1$ in the base cases become $\infty, 0,$ disjunctions ($\lor$) are replaced by ``$\min$'' and conjunctions ($\land$) by $+.$ Finally, the fourth recurrence relation becomes $\mathit{dp}[x_1, x_2, i, k, r] = \min\{w_i + \mathit{dp}[x_1, x_2, i - 1, k, r - 1, d], X\}.$
\end{proof}

\noindent We make use of \Cref{lemma:hiking-x-delete-min-weight-poly} to show the following.

\begin{restatable}{theorem}{maxsatweight}
\fauxsc{Hiking-Max-Satisfied-Weighted} is solvable in $O(n^8)$ time.
\end{restatable}
\begin{proof}[Proof Sketch] In the proof of \cref{th:hiking-max-satisfied-poly} for the unweighted case, we were looking for the largest $k$ such that there exists a size-$k$ subset $N' \subseteq N$ (corresponding to $k$ agents which we do not require to be satisfied) such that $(N \setminus N') \cup D_k$ admits a wonderful partition, where $D_{k}$ was a set of $k$ dummy agents happy with any group size. We then argued that we could relax to asking for a size-$k$ subset $N' \subseteq (N \cup D_k)$ such that $(N \cup D_k) \setminus N'$ admits a wonderful partition, which can be done using our polynomial-time algorithm for \fauxsc{Hiking-x-Delete-Min} in \cref{lemma:hiking-x-delete-poly}. This time, we follow a similar approach, except without binary search: we will try out all values $0 \leq k \leq n$ and ask for a size-$k$ subset $N' \subseteq (N \cup D_{k})$ such that $(N \cup D_k) \setminus N'$ admits a wonderful partition. For a fixed $k$, out of all abiding $N',$ we want one minimizing $\sum_{i \in N'}w_i.$ Such an $N'$ can be computed using our $O(n^7)$ algorithm for \fauxsc{Hiking-x-Delete-Min-Weight} in \cref{lemma:hiking-x-delete-min-weight-poly}. We do this for all values $0 \leq k \leq n$ and take the minimum-weight solution, adding an extra $O(n)$ factor, so the overall complexity is $O(n^8).$  
\end{proof}

% distance minimality %

\section{Single-Peaked Preferences Over Group Sizes}\label{sec:distanceMin}

We extend the binary version of the hiking problem by considering single-peaked preferences over group sizes. We assume that each agent $i$ has an ideal group size $s_i$ and the cost of agent $i$ if placed in a group of size $s$ is given by a cost function dependent on $s_i$ and $s$.

Given these ingredients, minimizing the social cost can be done in two variants:~a utilitarian variant and an egalitarian variant. In the utilitarian variant, the goal is to minimize the total cost of the agents, while in the egalitarian version we want the cost of the agent having the highest cost to be as low as possible, i.e., we replace summation with maximum. We are also interested in a variation of the problem where the hike organizers consider it reasonable to exclude at most $\alpha$ of the agents from the hike, the goal becoming to select the agents to remove and then to organize a hike with best social cost among the remaining agents.\footnote{The original problem can be seen to be the $\alpha = 0$ case.} Note that we assume that each excluded agent has a cost of $0$. This is different from assuming that excluded agents have to hike alone, which is covered by the case $\alpha=0$.

More formally, we assume that each agent $i \in N$ announces an ideal coalition size $s_i$; i.e., agent $i$ would be most happy when belonging to a coalition of size $s_i$. 
Moreover, given some cost function $\cost: N^2 \rightarrow \mathbb{R}$ agent $i$ incurs a cost of $\cost(s_i,s)$ if placed in a coalition of size $s \in [n]$ and a cost equal to $0$ if not participating in the hike. We assume $c$ to be monotone, i.e., $\cost(s_i,s) \leq \cost(s_i,s')$ if $s_i \leq s \leq s'$ or $s' \leq s \leq s_i$.  Notice that since agent $i$ incurs a disutility equal to $\cost(s_i,s)$, where $s_i$ is the most preferred size of $i$, the monotonicity condition on $\cost(\cdot, \cdot)$ implies that the preferences of agents are single-peaked w.r.t.\ the natural ordering. We refer to the next section for a formal definition.
Given a partition $\pi$, we recall that $\pi(i)$ denotes the coalition of agent $i$; if $i$ is not participating in the hike we write $\pi(i)=\bot$. Furthermore, by slight abuse of notation, for an agent $i$ we write $i \in \pi$ to indicate that agent~$i$ takes part in the hike; i.e., $\pi(i) \neq \bot.$

The utilitarian social cost of a partition $\pi$ is given by $\cost(\pi) = \sum_{i\in\pi}\cost(s_i,\modulus{\pi(i)})$ while the egalitarian social cost is given by $\cost(\pi) = \max_{i\in \pi}\cost(s_i,\modulus{\pi(i)})$. The goal is, therefore, to find a partition minimizing the social cost under the constraint $\modulus{\set{i \mid \pi(i)=\bot}} \leq \alpha$, or equivalently, $\modulus{\set{i \mid i\in\pi}} \geq n-\alpha$, where $\alpha$ is the maximum number of agents that are allowed to not participate in the hike. Without loss of generality, we assume that $s_1 \leq \ldots \leq s_n.$

We begin by proving two structural properties of optimal solutions that will allow us to greatly reduce the space of solutions that have to be considered, hence enabling us later to give efficient dynamic programming algorithms computing optimal partitions for both the utilitarian and the egalitarian settings. For the utilitarian social cost we need a mild assumption on the cost function $\cost$.

\begin{definition}
    A function $\cost: N^2 \rightarrow \mathbb{R}$ fulfills the \emph{quadrangle inequality} if and only if
    \[\cost(a,c) + \cost(b,d) \leq \cost(a,d) + \cost(b,c)\text{ for all } a \leq b \leq c \leq d.\]
    Analogously, it fulfills the \emph{reverse quadrangle inequality} if and only if 
    \[\cost(a,c) + \cost(b,d) \leq \cost(a,d) + \cost(b,c)\text{ for all } a \geq b \geq c \geq d.\]
\end{definition}
Note that quadrangle inequality and reverse quadrangle inequality are equivalent if $\cost$ is symmetric, i.e.\ $\cost(a,b) = \cost(b,a)$. Moreover, notice that if $cost(\cdot,\cdot)$ is the Euclidean distance on $\mathbb{R}$ or $\cost(a,b)=|b-a|^k$ for any $k\geq 1$, then it satisfies both aforementioned quadrangle inequalities.
The first observation that our approach will hinge upon is that it is enough to consider \emph{size-monotonic} partitions. 
\begin{definition}[Size-Monotonicity]
A partition $\pi$ is \emph{size-monotonic} if for any two agents $i, j \in \pi$, with $i < j$, it holds that $|\pi(i)| \leq |\pi(j)|.$
\end{definition}
Roughly speaking, there are optimal solutions where all participating agents with lower preferred coalition sizes belong to smaller coalitions than agents with higher preferred sizes. We show this fact in the following, proven in \cref{app:omitted}.
\begin{restatable}{lemma}{sizemonotonicity}
\label{lemma:size-monotonic} The following properties hold:
\begin{enumerate}[(a)]
    \item In the utilitarian setting, if $\cost(\cdot,\cdot)$ is monotone and fulfills quadrangle inequality and reverse quadrangle inequality, then there exists a size-monotonic optimal solution.
    \item In the egalitarian setting, if $\cost(\cdot,\cdot)$ is monotone, then there exists a size-monotonic optimal solution.
\end{enumerate}

\end{restatable}
\noindent Finally, we will show that it is enough to consider size-monotonic partitions which are additionally \emph{compact}. The latter is defined as follows:
\begin{definition}[Compactness]
A coalition $C$ is \emph{compact} if it is of the form $C = \{i, i + 1, \ldots, j\}$ for some $i \leq j.$ A solution $\pi$ is \emph{compact} if all coalitions $C \in \pi$ are compact.
\end{definition}
We now show that we can modify any optimal size-monotonic partition so that it is size-monotonic and compact, as follows. See \cref{app:omitted} for the proof. 
\begin{restatable}{lemma}{lemmacompact}\label{lemma:compact} 
There exists an optimal partition which is size-monotonic and compact.
\end{restatable}

\noindent With the above observations, we now know that it is enough to give an efficient algorithm to compute the best size-monotonic and compact solution. We do so in the following. In fact, our algorithm will only directly leverage compactness.\footnote{However, size-monotonicity is crucial in showing that considering compact solutions is enough.} To begin, for any two agents $i \leq j$ define $c(i, j)$ to be the social cost induced by agents $i, i + 1, \ldots, j$ when forming the coalition $\{i, i + 1, \ldots, j\}.$ In particular, $c(i, j) = \sum_{k = i}^{j}\cost(s_k,j - i + 1)$ in the utilitarian case, and similarly with summation replaced by maximum in the egalitarian case. With this definition in place, note that selecting the best compact solution $\pi$ with at most $\alpha$ agents not taking part in the hike amounts to selecting compact non-intersecting coalitions $C_1, \ldots, C_\ell$ 
where $C_k = \{a_k, a_k+1, \ldots, b_k -1, b_k\},$
such that $\sum_{k = 1}^{\ell}(b_k - a_k + 1) \geq n - \alpha$, 
and the sum/maximum of $c(a_1, b_1), \ldots, c(a_\ell, b_\ell)$ is minimized. Without loss of generality we can assume that $b_1 < a_2$, $b_2 < a_3$, \dots, $b_{\ell -1} < a_{\ell}$, i.e., we assume that the coalitions are sorted by index in increasing order. Before giving the actual algorithm, we note that, to get the best efficiency possible, we will need that the values $c(i, j)$, for all pairs $(i, j)$ with $i\leq j$, can be computed in total time $O(n^2)$ as a preprocessing step. We show this now.

\begin{lemma}\label{lemma:cost-precomputation} All values $c(i, j)$ for $i \leq j$ can be computed in total time $O(n^2).$
\end{lemma}
\begin{proof} We will compute the values separately for each value of $j - i.$ In particular, for each $0 \leq \ell < n$ we will compute all values $c(i, i + \ell)$ for $1 \leq i \leq n - \ell$ in linear time. To do this, for a fixed $\ell,$ note the contribution of agent $k$ to the $c$-values that it counts into is precisely $\cost(s_k,\ell + 1).$ Therefore, the values $[c(i, i + \ell)]_{1 \leq i \leq n - \ell}$ that we want to compute are aggregate queries over a sliding window of length $\ell + 1$ over the sequence $[\cost(s_k,\ell + 1)]_{1 \leq k \leq n}.$ Depending on the utilitarian/egalitarian goal, the aggregate can be either summation or maximum, but in either case, all the $n - \ell$ aggregates can be computed in linear time using standard sliding window techniques.
\end{proof}
We are now ready to present our algorithm. We construct a weighted directed acyclic graph $\calG$ corresponding to the problem instance, as follows. We are given a source node $s=(0,0)$ and a target node $ t= (n+1, *)$.
For the remaining vertices, we have one vertex for each pair $(i, j)$ with $1 \leq i \leq n + 1$ and $0 \leq j \leq \alpha$. Intuitively, vertex $(i, j)$ has the meaning ``agent $1 \leq i \leq n + 1$  is the next one to consider\footnote{Indeed, there is a ``dummy'' agent $n + 1$ signifying that there are no more agents to consider.} and so far we have excluded $0 \leq j \leq \alpha$ agents from the hike. 
The source $s$ is connected with a directed edge towards $(1,0)$ and such an edge has weight $0$, while $t$ is reachable from the node $(n+1, j)$, for each $0 \leq j \leq \alpha$, via an edge of weight $0$.
For the remaining edges, we add the following two types:
\begin{itemize}
\item We add a weighted edge $(i, j) \xrightarrow{0} (i + 1, j + 1)$ for all $1 \leq i \leq n$ and $0 \leq j < \alpha$. Intuitively, these correspond to excluding agent $i$ from the hike.
\item We add a weighted edge $(i, j) \xrightarrow{c(i, k)}(k + 1, j)$ for all $1 \leq i \leq k \leq n$ and $0 \leq j \leq \alpha,$. Intuitively, these correspond to adding a new coalition $C = \{i, i + 1, \ldots, k\}$ to the hike, incurring a cost of $c(i, k).$ 
\end{itemize}
\noindent The next lemma establishes how paths in $\calG$ correspond to compact solutions to our problem and its statement immediately follows by the construction of $\calG$.

\begin{lemma}\label{lemma:solutions-to-paths} There is a bijection from compact solutions to $s$-$t$ paths in $\calG.$ Moreover, the social cost of a compact solution is the cost of the associated path, defined as either the sum or the maximum of the costs of its constituent edges.
\end{lemma}
As a result, computing an optimal compact solution amounts to finding a minimum cost $s$-$t$ path in $\calG$; this gives us a polynomial time algorithm for computing an optimal solution. 

\begin{restatable}{theorem}{thmdistancemin}\label{thm:distance_min}
A hike with minimum social cost can be computed in time $O(n^2(\alpha + 1)).$
\end{restatable}
\begin{proof} By Lemma \ref{lemma:compact} it is enough to compute the best hike among compact solutions. To do so, we construct the graph $\calG$ corresponding to the problem instance. Subsequently, we compute an $s$-$t$ path in $\calG$ of minimum cost. This can be done in time linear in the size of the graph, as the graph is acyclic. Correctness is assured by \Cref{lemma:solutions-to-paths}. For the time bound, note that the number of vertices in the graph is $O(n(\alpha + 1))$ and the number of edges is $O(n^2(\alpha + 1)).$ Moreover, edge costs can be computed in constant time after $O(n^2)$ total precomputation by Lemma \ref{lemma:cost-precomputation}. Overall, we get a time complexity of $O(n^2(\alpha +1)).$
\end{proof}

% Egalitarian %

\section{Wonderful Partitions Versus Minimum Egalitarian Partitions}\label{sec:mineg}
So far we assumed each agent has an ideal group size, a peak, and there exists a cost function $cost(x, y)$ which expresses the cost that any agent having ideal group size $x$ incurs if placed in a coalition of size $y$.
More broadly, agents may express their disutility with any cost function, that is, for each agent $i$ there is a mapping $\cost_i: N\rightarrow \mathbb{R}$ simply expressing the cost agent $i$ incurs when assigned to a coalition of a certain size.  

In this section, we are interested in finding the minimum egalitarian cost achievable by any partition and we denote by \minEg\ the problem of computing this value.
In \Cref{sec:distanceMin}, we have already shown that whenever $cost_i(s)= cost(s_i, s)$, \minEg\ can be computed in $O(n^2)$. In what follows, we exploit the connection between \wonderful\ and \minEg\ delineating the tractability of the latter with respect to the properties of the cost functions.
First, we show a general connection between \wonderful\ and \minEg. 

\begin{proposition}\label{lemm:equivalence}
\wonderful\ and \minEg\ are polynomial time equivalent.
\end{proposition}
\begin{proof}
A \wonderful\ instance can be transformed into a \minEg\ instance by setting for each agent a cost function $cost_i$ where $cost_i(j)=1$ if agent $i$ does not approve coalition size $j$ and $cost_i(j)=0$, otherwise. Hence, the \wonderful\ instance
is a yes instance if \minEg\ is $0$ and it is a no instance, otherwise. 

If we have a \minEg\ instance, there are at most $n^2$ distinct values $cost_i(j)$. Assume there are $k$ distinct such values and let us sort them from the lowest to the highest, namely, $c_1< \ldots < c_k$. For each value~$c_h$, we can define a \wonderful\ instance where a coalition of size $j\in [n]$ is approved by agent~$i$ if and only if $cost_i(j)\leq c_h$. We can therefore determine if there exists a partition having egalitarian welfare of at most $c_h$ by solving \wonderful\ on the just described instance. 
Clearly, if there exists a partition having an egalitarian cost of at most $c$ then there exists a partition having an egalitarian cost of value at most $c'$ for each $ c \leq c '$. Conversely, if there is no partition having an egalitarian welfare of at most $c$ then there is no partition having an egalitarian cost of at most $c''$, for each $0\le c''\leq c$. Therefore, we can use binary search among the possible values $c_1, \ldots,  c_k$ to find the minimum value $c$ such that a partition having an egalitarian welfare of at most $c$ exists. This solves \minEg. 
\end{proof}

\begin{corollary}\label{cor:solvingEgalitarian}
    If \wonderful\ can be decided in time $T(n)$, then \minEg\ can be solved in time $O((n^2 + T(n))\cdot \log_2 n)$.
\end{corollary}

Given the polynomial time equivalence between these two problems, it follows that solving \minEg\ is in general computationally intractable
because \wonderful\ is NP-hard as soon as the approval sets are not intervals~\cite{Darmann_journal}. 

Nevertheless, there exists a special class of costs that can be solved using our dynamic programming approach for \wonderful\ described in Section~\ref{sec:polyDP}. Such a class is a generalization of what we discussed in \Cref{sec:distanceMin}: Namely, each agent has an ideal group size $s_i$, and the closer the coalition size is to $s_i$ the lower the cost.  We align to the Hedonic Games literature calling this property \emph{naturally single-peakedness}.\footnote{We say \emph{naturally} as general single-peakedness may be defined w.r.t.\ any fixed ordering of coalition sizes. In our setting, we consider cost functions that are single-peaked w.r.t.\ the natural ordering $1, \dots, n$.}

\begin{definition}
  A cost function $cost:\mathbb{N} \rightarrow \mathbb{R}_{\geq 0}$ is said to be \emph{naturally single-peaked} if there exists an ideal group size, a \emph{peak}, $p\in [n]$ such that $h<k \leq p$ or $h>k\geq p$ imply that $ cost(k) < cost(h)$ holds.  
\end{definition}

We observe that the reduction from \minEg\ to \wonderful\ described in \Cref{lemm:equivalence} produces an interval instance in the case of naturally single-peaked cost functions. With this, we obtain the following theorem. 

\begin{restatable}{theorem}{naturallysinglepeak}
\label{thm:naturally-single-peak}
\minEg\ for naturally single-peaked costs can be solved in time $O(n^5\log_2 n)$.
\end{restatable}
\begin{proof}
   We can use a similar idea as in the polynomial time reduction from \minEg\ to \wonderful, as described in \Cref{lemm:equivalence}. In \minEg\ we are looking for a partition that minimizes the cost of the agent having the highest disutility. Here, since we consider naturally single-peaked preferences, the disutility of an agent $i$ is decreased the closer the size of their assigned group $|\pi(i)|$ is to their ideal group size $s_i$ (the peak value of agent~$i$). We can now fix, for a value $c$, some distances $\Delta, \Delta'$ such that $cost_i(s)\leq c$ for each $s\in [s_i - \Delta, s_i + \Delta']$. Therefore, we define that any agent~$i$ approves their assigned group size if it is in the interval $[s_i - \Delta, s_i + \Delta']$; this defines the instance of \wonderfulIntervals\ to be solved for determining if there exists a partition of having an egalitarian cost of at most $c$.
   
   Now, applying \Cref{cor:solvingEgalitarian} and \Cref{thm:hiking-is-poly-time}, the statement follows.
\end{proof}

\noindent Notice that the computational complexity guaranteed by \Cref{thm:distance_min} is way more efficient than the one of \Cref{thm:naturally-single-peak}. However, the former holds true for very specific naturally single-peaked costs, while the latter establishes the tractability of \minEg\ whenever agents have naturally single-peaked costs.

\section{Conclusions}
We resolved an open problem posed a decade ago by Gerhard Woeginger by giving a polynomial-time algorithm via establishing a connection to a version of rectangle stabbing, and investi\-ga\-ted several interesting variants. We give a complete picture on the tractability of the Hiking Problem itself and show that maximizing the number of satisfied participants or deleting the minimal number such that the remaining participants admit a wonderful partition is polynomial time solvable. The tractability of both the original decision-, and the according optimization problems is crucially enabled by the existence of optimal solutions that exhibit simple and intuitive structural properties, fueling the algorithmic solutions based on Dynamic Programming. Last but not least, we employ our solution to efficiently compute a partition that maximizes
the egalitarian welfare for anonymous naturally single-peaked Hedonic Games. We note that our approach also works for general interval instances, that is, for a given permutation $\sigma$ of numbers $1,\dots,n$, intervals are defined over the numbers in order of the permutation, i.e., $\sigma(1), \dots ,\sigma(n)$. This extends our results from naturally single-peaked to general single-peaked cost functions. The problem of minimizing utilitarian cost for general single-peaked cost functions remains open.

\bibliography{bibliography}

\appendix
%%%%%%%%%%%%%%%%%%%%%%%%%%%%%%%%%%%%%%%%%%%%%%%%%%%%%%%%%%%%%%%%%%%%%%%%%%%%%%%%%%%%%%%%%%%%%%%%%%%%%%%%%%%%%%%%%%%%%%%%%%%%%%%%%%%%%%%%%%%%%%%%%%%%%%%%%%%%%%%%%%%%%%%%%%%%%%%%%%%%%%%%%

\section{Proofs Omitted From \Cref{sec:distanceMin}}\label{app:omitted}

\sizemonotonicity*
\begin{proof} 
We begin with property (a), i.e.\ the utilitarian case. Assume towards a contradiction that no optimal solution is size-monotonic. Let $\pi$ be an optimal solution minimizing the number of pairs $(i, j)$ of agents in $\pi$ such that $i < j$ and $|\pi(i)| > |\pi(j)|$. We call such pairs \emph{bad}. Since $\pi$ is not size-monotonic, at least one bad pair exists and let $(i, j)$ be such a bad pair. For ease of notation, denote $|\pi(i)|$ and $|\pi(j)|$ by $a$ and $b$, respectively. Consider the alternative solution $\pi'$ which is identical to $\pi$ except that $i$ and $j$ swap coalitions. Note that $\pi'$ has strictly fewer bad pairs than $\pi.$
Now, consider
\[
\cost(\pi) - \cost(\pi') = \cost(s_i,a) + \cost(s_j,b) - \cost(s_i, b) - \cost(s_j,a).
\] 
%$
If we could show that this quantity is non-negative, then we would get that $\pi'$ is also an optimal solution, but one having strictly less bad pairs, contradicting minimality. We will now show exactly this. Since $s_i \leq s_j$ and $b < a,$ it follows that we only need to consider the following six cases.

\paragraph*{Case 1a: $s_i \leq s_j \leq b < a$.}
 The inequality $\cost(s_i,a) + \cost(s_j,b) \geq \cost(s_i, b) + \cost(s_j,a)$ follows immediately from the quadrangle inequality.
\paragraph*{Case 2a: $s_i \leq b \leq s_j \leq a$.}
From quadrangle inequality we get $\cost(s_i,a) + \cost(b,b) \geq \cost(s_i,b) + \cost(b,a)$, i.e.\ the triangle inequality. By monotonicity we conclude
\begin{align*}
    \cost(s_i,a) + \cost(s_j,b) &\geq \cost(s_i,a) + \cost(b,b) \geq \cost(s_i,b) + \cost(b,a)\\
    &\geq \cost(s_i,b) + \cost(s_j,a).
\end{align*}
\paragraph*{Case 3a: $s_i \leq b < a \leq s_j$.}
By monotonicity we have $\cost(s_i,a)\geq \cost(s_i,b)$ and $\cost(s_j,b) \geq \cost(s_j,a)$. The inequality follows immediately.
\paragraph*{Case 4a: $b \leq s_i \leq s_j \leq a$.}
We use again monotonicity to observe $\cost(s_i,a) \geq \cost(s_j,a)$ and $\cost(s_j,b) \geq \cost(s_i,b)$. This obtains the desired result.
\paragraph*{Case 5a: $b \leq s_i \leq a < s_j$.}
We use monotonicity and reverse quadrangle inequality to get
\begin{align*}
    \cost(s_i,a) + \cost(s_j,b) &\geq \cost(a,a) + \cost(s_j,b) \geq \cost(s_j,a) + \cost(a,b)\\
    &\geq \cost(s_j,a) + \cost(s_i,b).
\end{align*}
\paragraph*{Case 6a: $b < a \leq s_i \leq s_j$.}
This case follows directly from the definition of the reverse quadrangle property. This finishes the proof for the utilitarian case.

Now, to tackle property (b), the egalitarian case, the argument remains similar, except that now we need to show that $\max \left\{\cost(s_i,a), \cost(s_j,b)\right\} \geq \max \left\{\cost(s_j,a), \cost(s_i,b)\right\}$. We again use the case distinction as above but we only need to assume monotonicity of $\cost$.

\paragraph*{Case 1b: $s_i \leq s_j \leq b < a$.}
 By monotonicity we have
 \begin{align*}
     \max \left\{\cost(s_i,a), \cost(s_j,b)\right\} = \cost(s_i,a) \geq \max\left\{\cost(s_i, b), \cost(s_j,a)\right\}
 \end{align*}
\paragraph*{Case 2b: $s_i \leq b \leq s_j \leq a$.}
We observe
 \begin{align*}
     \max \left\{\cost(s_i,a), \cost(s_j,b)\right\} \geq \cost(s_i,a) \geq \max\left\{\cost(s_i, b), \cost(s_j,a)\right\}.
 \end{align*}
\paragraph*{Case 3b: $s_i \leq b < a \leq s_j$.}
By monotonicity we have $\cost(s_i,a)\geq \cost(s_i,b)$ and $\cost(s_j,b) \geq \cost(s_j,a)$. The inequality follows immediately.
\paragraph*{Case 4b: $b \leq s_i \leq s_j \leq a$.}
We use again monotonicity to observe $\cost(s_i,a) \geq \cost(s_j,a)$ and $\cost(s_j,b) \geq \cost(s_i,b)$. This obtains the desired result.
\paragraph*{Case 5b: $b \leq s_i \leq a < s_j$.}
We get
\begin{align*}
    \max \left\{\cost(s_i,a), \cost(s_j,b)\right\} \geq \cost(s_j,b) \geq \max \left\{\cost(s_j,a),\cost(s_i,b)\right\}.
\end{align*}
\paragraph*{Case 6b: $b < a \leq s_i \leq s_j$.}
The last case can be obtained by
\begin{align*}
      \max \left\{\cost(s_i,a), \cost(s_j,b)\right\} = \cost(s_j,b) \geq \max \left\{\cost(s_j,a),\cost(s_i,b)\right\}.  
\end{align*}
This concludes the proof both for the utilitarian as well as the egalitarian case.
\end{proof}

\lemmacompact*

\begin{proof}
Let $\pi^*$ be an optimal size-monotonic partition, which, by \Cref{lemma:size-monotonic}, must exist. If partition~$\pi^*$ happens to be compact, then we are done. So, we can assume that partition~$\pi^*$ is not compact. Since compactness is violated, there exists a coalition $C\in \pi^*,$ such that $\min\set{h \mid h \in C}=i$ and $\max\set{h \mid h\in C}=j$, and there exists an agent $k\not\in C$ such that $i<k<j$. For agent $k$ there are only two options:
\begin{itemize}
\item[(i)] agent $k$ participates in the hike, i.e., we have $\pi^*(k) = C' \neq C$. By size-monotonicity it follows that $|\pi^*(i)| = |\pi^*(j)| = |\pi^*(k)|$; 
\item[(ii)] agent $k$ does not participate in the hike, i.e., $\pi^*(k) = \bot$.
\end{itemize}
We will now perform transformation steps to partition $\pi^*$ to ensure compactness without sacrificing size-monotonicity or optimality. Step~1 will remove compactness violations of type~$(i)$ and Step~2 will deal with type-$(ii)$-violations.

\textbf{Step 1:} Partition $\pi^*$ contains coalitions of certain sizes.
Given a coalition size $q\in [n]$, we rearrange the coalitions of size $q$ in $\pi^*$ in the following way:
Let $C_1,C_2,\dots,C_\ell$ be the coalitions in $\pi^*$ 
of size $q$,
and let $A_q = \bigcup_{1\leq i \leq \ell}C_i = \{i_1,i_2,\dots,i_{q\cdot\ell}\}$ be the set of agents that are in a coalition of size~$q$ in partition~$\pi^*$. Moreover, we assume that the agents in $A_q$ are sorted in increasing order of their ideal coalition sizes, i.e., we assume that  $s_{i_1} \leq s_{i_2} \leq \dots \leq s_{i_{q\cdot\ell}}$ holds. We now create $\ell$ many new coalitions of size~$q$ by reassigning the agents in $A_q$ as follows: the first $q$ agents $i_1,\dots,i_q$ are assigned to coalition~$C_1^*$, the next $q$ agents $i_{q+1},\dots,i_{2q}$ are assigned to coalition~$C_2^*$, and so on. The last $q$ agents $i_{q\cdot(\ell-1)+1}, \dots, i_{q\cdot\ell}$ are then assigned to coalition~$C_\ell^*$. Then we modify partition~$\pi^*$ by replacing the coalitions $C_1,\dots,C_\ell$ with the coalitions $C_1^*,\dots,C_\ell^*$. Note that after this replacement, partition $\pi^*$ is still size-monotonic, since the respective coalition size stays the same for every agent in $A_q$. 
We iterate this replacement procedure for 
each coalition size $q$.
We end up with a modified partition $\pi^*$ that is still size-monotonic. furthermore, by construction,  since we reassigned agents that are in coalitions having the same sizes in increasing order of their ideal coalition sizes, we cannot have compactness violations of type $(i)$ in the final partition $\pi^*$ at the end of Step 1.   

\textbf{Step 2:} We will remove type-$(ii)$-violations of compactness one by one. For keeping track of our progress, we consider the following measure.
Given a coalition $C$, the \emph{diameter of $C$} is defined as $\text{diam}(C)= \max\set{h \mid h\in C} - \min\set{h \mid h\in C}$. Moreover, the \emph{diameter of a partition $\pi$} is $\text{diam}(\pi)= \sum_{C\in\pi}\text{diam}(C)$. 

Assume that there is a coalition $C\in \pi^*$ with $\min\set{h \mid h \in C}=i$ and $\max\set{h \mid h\in C}=j$, that there exists an agent $k\not\in C$ such that $i<k<j$, and that we have $\pi^*(k) = \bot$. There are two cases, depending on agent $k$'s ideal coalition size $s_k$. If $s_k \geq |C|$, then we will change partition $\pi^*$ by swapping agents $k$ and $j$, i.e., agent $j$ will be excluded from the hike and agent $k$ will be assigned to coalition $C$ instead. If $s_k < |C|$, then we swap agents $i$ and $k$. Since $s_i \leq s_k \leq s_j$, none of those swaps can increase the social cost of partition $\pi^*$, i.e., partition $\pi^*$ stays optimal, and the partition $\pi^*$ stays size-monotonic. Moreover, in both cases the diameter of partition $\pi^*$ is strictly decreased. This implies that after finitely many such exchange steps, this process must stop and all type-$(ii)$-violations are resolved.  

It remains to show that the exchanges done in Step 2 do not create new type-$(i)$-violations of compactness. We show this via a proof by contradiction. Assume that for some coalition $C \in \pi^*$ with $\min\set{h \mid h \in C}=i$ and $\max\set{h \mid h\in C}=j$ we have exchanged agent $k$ with agent $j$ in Step 2 and this creates a new type-$(i)$-violation of compactness, that is, there is some agent $k'$ with $i<k'<j'$, where $j'=\max\set{h \mid h\in C}$ after the exchange. However, also before the exchange of $j$ and $k$, we had that $i<k'<j'<j$ holds, which implies that agent $k'$ already was a type-$(i)$-violation of compactness. This contradicts that a new type-$(i)$-violation was introduced. The argument for the case where agents $i$ and $k$ are exchanged is completely analogous.  

Thus, by first performing Step 1 to remove all type-$(i)$-violations of compactness and then performing Step 2 to remove all type-$(ii)$-violations of compactness we will eventually transform the optimal size-monotonic partition $\pi^*$ into an optimal partition that this still size-monotonic but also compact. 
\end{proof}

\section{Hardness for Approval Sets of Size 2}\label{app:hardness}

We have shown that whenever the approval set of each agent is an interval, solving \wonderful\ is possible in polynomial time. Clearly, whenever all agents have approval sets of size $1$, the existence of a wonderful partition is polynomial-time decidable as well: This can also be seen as a special case of interval instances (although it can also be solved directly, as explained in the introduction). In general, the approval sets are not necessarily intervals, and without any assumption on the structure of the approval sets, the \wonderful\ problem is NP-complete even if the size of each approval set is at most $3$. This follows from the hardness proof of Woeginger~\cite{woeginger2013core} and the fact that \fauxsc{Exact Cover By 3-Sets} is NP-hard (see below for the definition).\footnote{In fact, as a curiosity, even if every element of the ground set appears in exactly three triples and any two triples overlap in at most one element~\cite{Gonzalez85}.} Darmann et al.~\cite{Darmann_journal} establish an even stronger version of this result for the case where all approval sets are of size at most $2$.

In this section, we precisely map the boundary of tractability of \wonderful\ with respect to the approval set size. While the case with approval set size 1 can be solved in polynomial time, we now show NP-completeness if the approval sets have size \emph{exactly} 2. In particular, we will show the following:
\begin{theorem}\label{thm:hardness2approval}
Deciding \wonderful\ is NP-complete, even if the approval sets are of size \emph{exactly} 2.
\end{theorem}

We first observe that \wonderful\, when restricted to the case of approval sets of size exactly~2, is equivalent to a graph-theoretic problem we call \fauxsc{Orientation}. For this, we write $d^+(v)$ to denote the in-degree of a vertex $v \in V$ in a directed graph $\Vec{G} = (V, \Vec{E})$. 

\begin{tcolorbox}
\textbf{\fauxsc{Orientation}}\\
\textbf{Input}: An undirected graph $G = (V, E)$ with $V = [n],$ admitting parallel edges but no self-loops. \\
%where $V=\set{1, \dots, n}=[n]$, and  $G$ admits parallel edges.\\
\textbf{Question}: Does there exist an orientation of the edges 
%of $G$,
$\Vec{G}$ such that $d^+(i) \equiv 0 \pmod i$ for each $i\in V$?
\end{tcolorbox}

We now show that the two problems are equivalent. 

Indeed, an instance of \wonderful\ where all agents approve exactly two sizes can be transformed into an equivalent instance of \fauxsc{Orientation} by representing each agent approving sizes $i, j \in [n]$ by an undirected edge $(i, j)$. Orienting this edge to either node~$i$ or node~$j$ models that the respective agent is part of a partition of size $i$ or $j$, respectively. If $d^+(i) = k\cdot i$ for some $k\geq 0$ then this means that $k$ partitions with size $i$ will be created.
Conversely, an instance of \fauxsc{Orientation} with $m$ edges can be transformed to an instance of \wonderful\ with $m$ agents, where each edge $(i, j)$ corresponds to an agent with approval set $\set{i,j}.$

Using the equivalence between \wonderful\ and \fauxsc{Orientation}, we will prove~\Cref{thm:hardness2approval} by reducing \fauxsc{Exact Cover By 3-Sets} (\fauxsc{X3C}) to \textsc{Orientation}. \fauxsc{X3C} is well-known to be NP-hard~\cite{garey1979computers} and it is defined as follows:

\begin{tcolorbox}
\textbf{\fauxsc{Exact Cover By 3-Sets (X3C)}} \\
\textbf{Input}: A ground set $X=\set{x_1, \dots, x_{3k}}$ and a collection $\mathcal{C}$ of 3-element subsets (triples) of $X$. \\
\textbf{Question}: Does there exist a subset $\mathcal{C}' \subseteq \mathcal{C}$ such that $\bigcup_{C\in\mathcal{C}'}C=X$ and $\modulus{\mathcal{C}'}=k$, i.e., $\mathcal{C}'$ is an exact cover of $X$?
\end{tcolorbox}

\begin{proof}[Proof of \Cref{thm:hardness2approval}]
Let us assume that $\modulus{\mathcal{C}}= q$ and let us enumerate all the triples as $C_1, \dots, C_q$. For simplicity, we fix  $p=3k$. Notice that any element $x$ in the ground set can be contained in at most ${p-1 \choose 2} < p^2$ triples. Let us denote by $\delta$ the highest number of such occurrences. Our reduction is as follows:
\begin{itemize}
    \item For each element in the ground set $x_i\in X$, for $i\in [p]$, we create an \emph{element gadget} which consists of nodes $i\cdot \delta$ and $i\cdot\delta +1$, respectively. Such two nodes are connected by $i\cdot\delta -1 $ many edges. We may refer to these edges as element-($x_i$)-edges. Moreover, we call node $i\cdot\delta$ the element-($x_i$)-node.

    \item For each triple $C_j \in \mathcal{C}$, for $j\in [q]$, we create a \emph{triple gadget} which consists of two nodes having value $p\cdot\delta + 4j$  and $p\cdot\delta + 4j -3$, respectively. Such two nodes are connected by $p\cdot\delta + 4j -3$ many edges. We will refer to these edges as triple-($C_j$)-edges. Moreover, we call node $p\cdot\delta + 4j$ the triple-($C_j$)-node.

    \item Finally, for each $x_i\in X$ and $C_j\in\mathcal{C}$ such that $x_i\in C_j$, we connect the nodes $i\cdot\delta$ and $p\cdot\delta + 4j$ with one edge.
\end{itemize}
Clearly, the reduction is polynomial. Moreover, the constructed graph is well-defined (there is no overlap in gadgets corresponding to different compounds). See Figure~\ref{fig::gadgets} for an example.
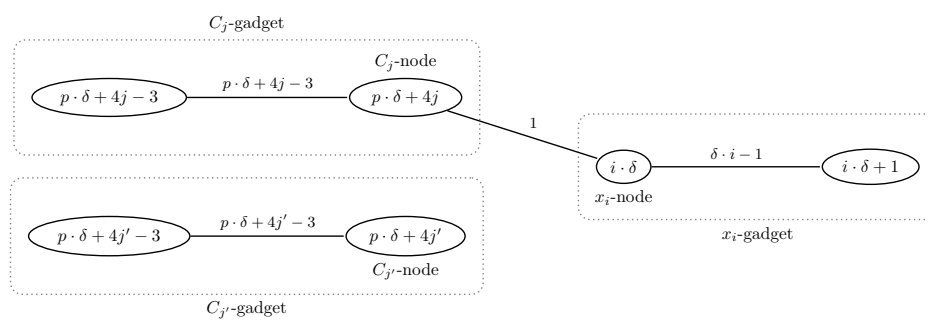
\begin{figure}[h]
\centering
%\hspace*{-1.5cm}
  \scalebox{0.65}
  {
  \begin{tikzpicture}[-,>=stealth',shorten >=1pt,auto,node distance=2cm,
  thick,main node/.style={ellipse,draw,font=\sffamily\bfseries}]
  \node[main node] (1) {$p\cdot \delta +4j -3$};
  \node[main node] (2) [right of=1, xshift=4cm, label={above: $C_j$-node}] {$p\cdot \delta +4j$};
  \node[main node] (3) [below right of=2, xshift=3cm, label={below: $x_i$-node }] {$i\cdot\delta$};
    \node[main node] (8) [right of=3 , xshift=3cm] {$i\cdot\delta + 1$};
  \node[main node] (5) [below left of=3, xshift=-3cm, label={below: $C_{j'}$-node}] {$p\cdot \delta +4j'$};
  \node[main node] (6) [left of=5, xshift=-4cm]{$p\cdot \delta +4j' -3$};
  \node[draw, rounded corners=7pt, dotted, gray, fit=(1) (2),inner ysep=22pt,inner xsep=10pt,label={above: $C_j$-gadget}] {};
  \node[draw, rounded corners=7pt, dotted, gray, fit=(5) (6),inner ysep=22pt,inner xsep=10pt, label={below: $C_{j'}$-gadget}] {};
  \node[draw, rounded corners=7pt, dotted, gray, fit=(3) (8),inner ysep=20pt,inner xsep=10pt,label={below: $x_i$-gadget }] {};
  \path[every node/.style={font=\sffamily\small}]
  (1) edge node {$p\cdot \delta +4j -3$} (2)
%  (2) edge node {$(4 i +n-3)\cdot \lfloor \frac{n}{3} \rfloor$} (3)
  (2) edge node {$1$} (3)
  (6) edge node {$p\cdot \delta +4j' -3$} (5)
%  (5) edge node {$(4 j +n-3)\cdot \lfloor \frac{n}{3} \rfloor$} (7)
 % (3) edge node {$1$} (5)
  (3) edge node {$\delta\cdot i -1$} (8)
  ;
  \end{tikzpicture}
  }
  \caption{Gadgets -- On the left, above (resp. below) the triple gadget for $C_j$ (resp. $C_{j'}$); on the right, the element gadget for $x_i$. The picture shows the set-up of the gadgets if $x_i \in C_j$ but $x_i \notin C_{j'}$. Labels on edges represent the number of parallel edges between the two nodes.}
  \label{fig::gadgets}
 \end{figure}

The idea is to determine the exact cover by means of the orientation of edges connecting elements and triple gadgets. Specifically, whenever a triple $C_j$ is in the covering set, the three edges connecting the gadget $C_j$ with the corresponding element gadgets
must all be oriented towards the corresponding element nodes. In turn, if $C_j$ is not in the covering set, the orientation of these edges must be towards the triple-($C_j$)-node.

With this, it follows that whenever an exact cover exists, an orientation of $G$ exists as well. In particular, for any triple $C_j$ in the covering set we can orient all edges incident to the triple-($C_j$)-node away from the triple-($C_j$)-node, and in the opposite direction, otherwise. Moreover, every element edge is directed towards the corresponding element node.

It remains to show that whenever an orientation exists, an exact cover exists as well.
The rest of this proof is established by the following observations:
\begin{enumerate}
    \item In any feasible orientation, the element edges are always oriented toward the element node. Otherwise, the orientation will not be feasible since the non-element node endpoint of every element edge has a label that is higher than the number of its incident edges.
    \item Denote by $t_i$ the number of triples containing $x_i$, the element-($x_i$)-node has $i\cdot\delta-1+t_i$ incident edges. By (1), the in-degree in a feasible orientation is at least $i\cdot\delta-1$. Since $t_i\leq \delta$, in a feasible orientation of $G$ the in-degree of $i\cdot\delta$ is exactly its value. As a consequence, there is only one incoming edge from triple gadgets, and therefore each element is covered by exactly one triple.
    \item Last but not least, in a feasible orientation, a triple cannot be only partially used. Specifically, for a triple $C_j$ with triple-($C_j$)-node $v$ either all edges incident to $v$ must be all oriented to $v$ or all away from $v$. This is ensured by the fact that in any feasible orientation the in-degree of a triple $C_j$ is either $0$ or $p\cdot\delta +4j$. \qedhere
\end{enumerate}
\end{proof}

%%%%%%%%%%%%%%%%%%%%%%%%%%%%%%%%%%%%%%%%%%%%%%%%%%%%%%%%%%%%%%%%%%%%%%%%%%%%%%%%%%%%%%%%%%%%%%%%%%%%%%%%%%%%%%%%%%%%%%%%%%%%%%%%%%%%%%%%%%%%%%%%%%%%%%%%%%%%%%%%%%%%%%%%%%%%%%%%%%%%%%%%%%%%%%%%%%%%%

\end{document}